\documentclass[american,aps,pra,reprint,floatfix,nofootinbib,superscriptaddress,longbibliography]{revtex4-1}
\usepackage[unicode=true,pdfusetitle, bookmarks=true,bookmarksnumbered=false,bookmarksopen=false, breaklinks=false,pdfborder={0 0 0},backref=false,colorlinks=false]{hyperref}
\hypersetup{colorlinks,linkcolor=myurlcolor,citecolor=myurlcolor,urlcolor=myurlcolor}
\usepackage{graphics,epstopdf,graphicx,amsthm,amsmath,amssymb,mathptmx,braket,colortbl,color,bm,framed,mathrsfs}
\usepackage[T1]{fontenc}
\usepackage[up]{subfigure}
\usepackage{tikz}

\definecolor{myurlcolor}{rgb}{0,0,0.9}

\newcommand{\proj}[1]{| #1\rangle\!\langle #1 |}

\newcommand{\inner}[2]{\langle #1 , #2\rangle}

\DeclareMathOperator{\trace}{Tr}
\DeclareMathOperator{\e}{e}
\newcommand{\Ptr}[2]{\trace_{#1}\Pa{#2}}
\newcommand{\Tr}[1]{\Ptr{}{#1}}

\newcommand{\Pa}[1]{\left[#1\right]}

\newcommand{\norm}[1]{\left\lVert #1 \right\rVert}

\theoremstyle{plain}
\newtheorem{thm}{Theorem}
\newtheorem{lem}[thm]{Lemma}
\newtheorem{prop}[thm]{Proposition}

\newcommand*{\myproofname}{Proof}

\def\ot{\otimes}

\def\real{\mathbb{R}}

\def \diag {\mathrm{diag}}

\DeclareMathAlphabet{\mathcal}{OMS}{cmsy}{m}{n}

\makeatother

\begin{document}

  \author{Kaifeng Bu}
 \email{kfbu@fas.harvard.edu}
\affiliation{Department of Physics, Harvard University, Cambridge, Massachusetts 02138, USA}

\author{Dax Enshan Koh}
  \email{dax\_koh@ihpc.a-star.edu.sg}
\affiliation{Institute of High Performance Computing, Agency for Science, Technology and Research (A*STAR), 1 Fusionopolis Way, \#16-16 Connexis, Singapore 138632, Singapore}

\author{Lu Li}
\affiliation{Department of Mathematics, Zhejiang Sci-Tech University, Hangzhou, Zhejiang 310018, China}
\affiliation{School of Mathematical Sciences, Zhejiang University, Hangzhou, Zhejiang 310027, China}

\author{Qingxian Luo}
\affiliation{School of Mathematical Sciences, Zhejiang University, Hangzhou, Zhejiang 310027, China}
\affiliation{Center for Data Science, Zhejiang University, Hangzhou Zhejiang 310027, China}

\author{Yaobo Zhang}
\affiliation{Zhejiang Institute of Modern Physics, Zhejiang University, Hangzhou, Zhejiang 310027, China}
\affiliation{Department of Physics, Zhejiang University, Hangzhou Zhejiang 310027, China}

\title{Effects of quantum resources on the statistical complexity of quantum circuits}

\begin{abstract}

We investigate how the addition of quantum resources changes the statistical complexity of quantum circuits by utilizing the framework of quantum resource theories. Measures of statistical complexity that we consider include the Rademacher complexity and the Gaussian complexity, which are well-known measures in computational learning theory that quantify the richness of classes of real-valued functions. We derive bounds for the statistical complexities of quantum circuits that have limited access to certain resources and apply our results to two special cases: (1) stabilizer circuits that are supplemented with a limited number of T gates and (2) instantaneous quantum polynomial-time Clifford circuits that are supplemented with a limited number of CCZ gates. We show that
the increase in the statistical complexity of a quantum circuit when an additional quantum channel is added to it is upper bounded by the free robustness of the added channel. Finally, we derive bounds for the generalization error associated with learning from training data arising from quantum circuits.

\end{abstract}

\maketitle

\section{Introduction}

Quantum machine learning, which aims to harness the power of quantum computing to perform machine learning tasks, has attracted considerable interest in recent years \cite{lloyd2013quantum,wittek2014quantum, BiamonteNature17,ciliberto2018quantum,Dunjko2018}. This interest is accompanied by the hope that quantum algorithms can outperform their classical counterparts at solving certain machine learning problems. This hope is fuelled, in part, both by the observation that quantum computers are capable of efficiently producing patterns in data that classical computers are believed to not be able to produce efficiently \cite{bremner2010classical,
aaronson2011computational,
BremnerPRL2016,
dalzell2020many} and by the proposal of quantum algorithms with a provable exponential speedup over known classical algorithms that may be adapted for use as subroutines in certain machine learning algorithms. An example of such an algorithm is the Harrow-Hassidim-Lloyd (HHL) algorithm \cite{Harrow08} for solving linear systems of equations, which has been applied to various machine learning problems, like recommendation systems \cite{kerenidis2016quantum}, support vector machines \cite{Rebentrost14}, principal component analysis \cite{Lloyd14}, etc.

Central to many quantum machine learning algorithms is the need to train quantum variational circuits to perform certain tasks. These circuits are the central building block used in variational quantum algorithms, which have been described as a leading candidate for achieving a practical quantum advantage using noisy intermediate-scale quantum (NISQ \cite{preskill2018quantum,bharti2021noisy}) devices \cite{cerezo2020variational}. Examples of variational quantum algorithms include the variational quantum eigensolver (VQE) for quantum chemistry \cite{peruzzo2014variational,cao2019quantum,Lee2019}, the quantum approximate optimization algorithm (QAOA) for optimization \cite{farhi2014quantum}, and the quantum neural network (QNN) that generalizes the classical neural network  \cite{farhi2018classification,mitarai2018quantum,Schuld2019,havlivcek2019supervised,Sharma2020,beer2020training}.

While quantum circuits are believed to provide an advantage over their classical counterparts, not all of them are capable of doing so. There are several well-known restricted classes of quantum circuits that can be shown to be efficiently simulable by a classical computer. These include the stabilizer circuits \cite{gottesman1997heisenberg}, the matchgate circuits \cite{Valiant02,Jozsa08} as well as these circuits augmented with various supplementary resources \cite{jozsa2014classical,
koh2015further,
brod2016efficient, bu2019efficient,
hebenstreit2020computational}. These classical simulation results show that if one hoped to outperform classical algorithms at a given machine learning task, it is necessary to utilize resources outside these classically simulable restricted classes of quantum circuits.

A key step in both classical and quantum machine learning is the building of learning models based on training data. Here, the power of a learning model depends on its statistical complexity, i.e., its ability to fit functions, which has been quantified by various measures. These measures include the Vapnik-Chervonenkis (VC) dimension \cite{Vapnik71,Vapnik82} (which has been used to determine the sample complexity of
PAC learning \cite{Blumer89} and classical neural networks \cite{Harvery17}), the metric entropy (also known as 
covering number) \cite{tikhomirov1993varepsilon}, the Rademacher complexity \cite{Bartlett03} (which has been studied in the context of classical neural networks \cite{Neyshabur15,Bartlett17,Neyshabur2017,Golowich18a})  and the Gaussian complexity \cite{Bartlett03}.

In building learning models using quantum circuits, various quantum resources, or \textit{quantum effects}, are typically at play. These include magic \cite{Veitch14,howard_2017,wang2019quantifying}, entanglement \cite{Horodeck09,Plenio07} and coherence \cite{aberg2006quantifying,baumgratz14,Streltsov17}. But how do changes in the amounts of these quantum resources affect the statistical complexity of these quantum-circuit-based learning models? In recent work \cite{bu2021on}, we partially addressed this question by focusing on a specific resource, namely the resource of magic. In particular, we utilized the $(p,q)$ group norm to quantify the amount of magic in quantum circuits and showed how the statistical complexity of the quantum circuit scales with the depth and width of the circuit and the amount of magic it contains. In this work, 
we extend our previous results and address the above question more generally by considering the Rademacher and Gaussian complexities as measures of model complexity and  utilizing the framework of general resource theories \cite{COECKE16,chitambar_2019}, which offers a powerful paradigm for the quantification and operational interpretation of quantum effects \cite{howard_2017}.

In a quantum resource theory, quantum channels are categorized as being either a \textit{free} channel or a \textit{resource} channel. Free channels are those that are available or inexpensive and resource channels are those that are limited or expensive to use. 
In this work, we consider quantum-circuit-based learning models in the following two contexts: (1) quantum circuits with access to only a restricted set of channels $\mathcal O$, and (2) quantum circuits with access to a restricted set of channels $\mathcal O$ together with an additional resource channel $\Psi \in \mathcal O$ (for example, we could take $\mathcal O$ to be the set of stabilizer circuits and $\Psi$ to be the $T$ gate). We show that by adding a resource channel to a set of free channels, the Radamacher and Gaussian complexities are increased by an amount that is bounded by the free robustness of the resource channel multiplied by the number of times the channel is used. Using this result, we derive an upper bound on the generalization error associated with learning from the training data arising from such circuits.

\section{Preliminaries}

\subsection{Quantum-generated function classes}

Consider an $n$-qubit quantum circuit that implements a quantum channel\footnote{Subsequently, we will identify the circuit with the channel it implements and denote both by $\Phi$.} $\Phi$. For example, $\Phi=\Phi(\theta)$ could represent a parametrized quantum circuit with gates parametrized by the parameters $\theta \in \mathbb R^\alpha$ (for an example, see Fig.~\ref{fig1}).
Let $\vec{x}\in \mathbb{F}^n_2$ be an $n$-bit input string (for example, $\vec{x}$ could be the binary representation of a collection of pixel values of an image of a handwritten digit). By Born's rule, if we feed the computational basis state $\ket{\vec{x}}$ into the circuit $\Phi$ and make a measurement in the computational basis, the probability of measuring $\vec{y}\in \mathbb{F}^n_2$ is given by
\begin{align}
\label{def:fun}
    p_{\Phi,\vec{x}}(\vec{y})=f_\Phi(\vec x,\vec y):=\Tr{\Phi(\proj{x})\proj{y}}. 
\end{align}
where $f_\Phi:\mathbb{F}^n_2\times \mathbb{F}^n_2\to [0,1]$ is a real-valued function induced by the channel $\Phi$ that maps input-output pairs $(\vec x,\vec y)$ to probability values. Let $\Omega$ be a set of quantum channels. We define the function class $\mathcal F(\Omega)$ as follows:
\begin{align}
\label{eq:functionClass}
    \mathcal F(\Omega) = \{
    f_\Phi|
    \Phi \in \Omega
    \}.
\end{align}
\begin{figure}[t]
  \center{\includegraphics[width=8cm]  {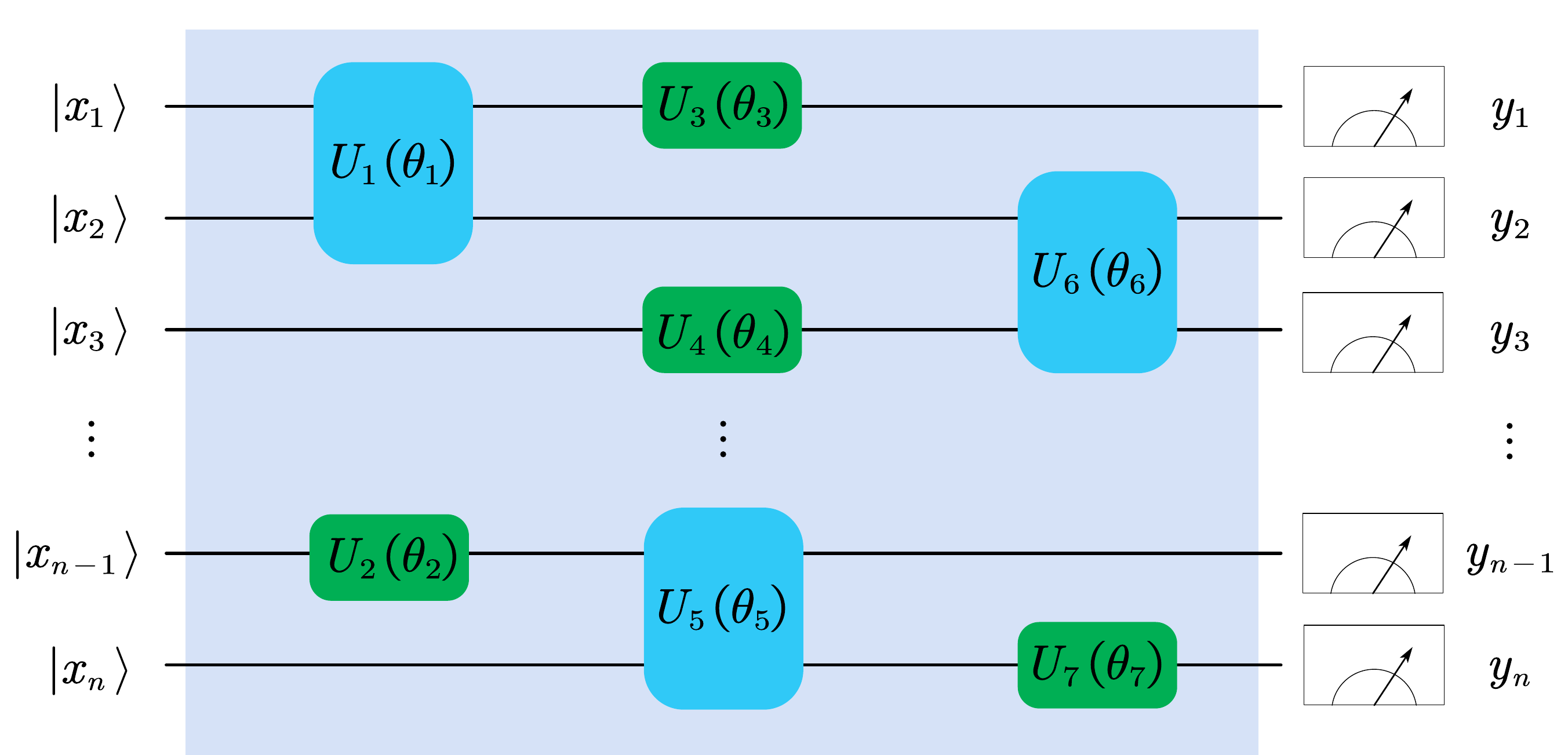}}     
  \caption{
  An example of a parametrized quantum circuit with parametrized gates.}
  \label{fig1}
 \end{figure}

\subsection{Statistical complexity}

We now introduce the Rademacher and Gaussian complexities \cite{Bartlett03}, which quantify the richness of sets of real-valued functions and can be used to provide bounds for the generalization error associated with learning from training data.
Let $\mathcal G$ be a set of real-valued functions and let $S=(z_1,\ldots,z_m) \in \mathbb R^m$ be a set of $m$ samples. The \textit{(empirical) Rademacher complexity} of $\mathcal G$ with respect to $S$ is
\begin{align}
    \hat{R}_S(\mathcal{G})=
\mathop{\mathbb{E}} \limits_{\substack{\epsilon_1,\ldots,\epsilon_m 
\\
\sim \mathrm{Rad}}} \left[
\sup_{g\in \mathcal{G}}
\frac{1}{m}\sum^m_{i=1}\epsilon_i g(z_i)
\right],
\end{align}
where the expectation is taken over i.i.d.~Rademacher random variables, i.e., $\epsilon_i \sim \mathrm{Rad}$ for each $i \in \{1,\ldots,m\}$. Recall that the Rademacher random variable $X$ has probability mass function
\begin{align}
    \mathrm{Pr}(X=k) = \begin{cases}
    1/2 & k \in \{-1,1\}, \\
    0 & \mbox{otherwise}.
    \end{cases}
\end{align}

Similarly, the \textit{(empirical) Gaussian complexity} of $\mathcal G$ with respect to $S$ is
\begin{align}
    \hat{G}_S(\mathcal{G})=
\mathop{\mathbb{E}} \limits_{\substack{\epsilon_1,\ldots,\epsilon_m 
\\
\sim \mathcal N(0,1)}} \left[
\sup_{g\in \mathcal{G}}
\frac{1}{m}\sum^m_{i=1}\epsilon_i g(z_i)
\right],
\end{align}
where the expectation is taken over the i.i.d.~random Gaussian variables with zero mean and unit variance, i.e., $\epsilon_i\sim \mathcal{N}(0,1)$ for each $i \in \{1,\ldots,m\}$.

Note that the empirical Rademacher and Gaussian complexities depend on the samples $S=(z_1,\ldots,z_m)$. By averaging over samples $S$ taken from a product distribution $D^m$, we obtain the \textit{expected Rademacher} and \textit{Gaussian complexities}:
\begin{align}
R_D(\mathcal{G})&=\mathop{\mathbb{E}}\limits_{S\sim D^m}\left[\hat{R}_S(\mathcal{G})\right],\\
G_D(\mathcal{G})&=\mathop{\mathbb{E}}\limits_{S\sim D^m}\left[\hat{G}_S(\mathcal{G})\right].
\end{align}
In the rest of the main text, we will focus on the Rademacher complexity; similar results hold for the Gaussian complexity, which we relegate to Appendix \ref{sec:GaussianComplexity}.

\subsection{Statistical complexity in the quantum resource theory framework}

Quantum resource theories are characterized by a restricted set of channels, called \textit{free channels}, which map free states to free states; any channel that is not a free channel is called a \textit{resource channel}. 
Let $\mathcal O$ be a set of $n$-qubit free channels and let $\Psi \notin \mathcal O$ be an $n$-qubit resource channel. Define $\mathcal{O}_{\Psi}:=\mathcal{O}\cup \set{\Psi}$ to be the class of channels formed by appending $\Psi$ to $\mathcal O$. In addition, to take into account the case where the resource channel is used more than once, for each $k\in \mathbb Z^+$, define 
\begin{align}
\mathcal{O}^{(k)}_{\Psi}
=\bigg\{\prod^l_{i=1}\Phi_i \bigg| 
l=\mathsf{poly}(n) ;
\Phi_i\in \mathcal{O}_{\Psi}
\ \forall i \in \{1,\ldots, l\}; \nonumber\\
\text{ and at most $k$ of the $\Phi_i$'s are $\Psi$}\bigg\}.
\label{eq:defOpsiK}
\end{align}
It is easy to see that the above sets form a nested hierarchy
\begin{align}
\label{eq:hierarchySets}
    \mathcal O \subset \mathcal{O}_{\Psi}\subset \mathcal{O}^{(1)}_{\Psi} \subset \mathcal{O}^{(2)}_{\Psi}\subset
    \ldots
    \subset
    \mathcal{O}^{(k)}_{\Psi}\subset
    \mathcal{O}^{(k+1)}_{\Psi}
    \subset \ldots
\end{align} 

In this work, we will be interested in the statistical complexities of the function classes $\mathcal{F}(\Omega)$ formed by taking $\Omega$ to be the sets in the nested hierarchy in Eq.~\eqref{eq:hierarchySets}, where $\mathcal{F}(\cdot)$ is given by Eq.~\eqref{eq:functionClass}.

\section{Results}

\subsection{Statistical complexity bounds}

We first consider the Rademacher complexity of $\mathcal{F}(\mathcal{O}_{\Psi})$.
\begin{thm}\label{thm:main1}
Given $m$ independent samples $S=(\vec{z}_1,\ldots,\vec{z}_m)$ and a resource channel $\Psi$, the Rademacher complexity of $\mathcal{F}(\mathcal{O}_{\Psi})$ is bounded as follows:
\begin{eqnarray}
\hat{R}_S(\mathcal{F}(\mathcal{O}))
\leq
\hat{R}_S(\mathcal{F}(\mathcal{O}_{\Psi}))
\leq (1+\gamma(\Psi)) \hat{R}_S(\mathcal{F}(\mathcal{O})),
\end{eqnarray}
where $\gamma(\Psi)$ is the free robustness of $\Psi$ with respect to the set $\mathcal{O}$, defined as
\begin{eqnarray*}
\gamma(\Psi):=\min\left\{\lambda \bigg|
\exists\Phi\in \mathrm{Conv}(\mathcal{O}):
\frac{\Psi+\lambda\Phi}{1+\lambda}\in \mathrm{Conv}(\mathcal{O})\right\}.
\end{eqnarray*}
 Therefore, for any probability distribution $D$ on the sample space, if each sample $\vec{z}_i$ is chosen independently 
according to $D$ for $i=1,\ldots,m$, we have 

\begin{eqnarray}
R_D(\mathcal{F}(\mathcal{O}))
\leq
R_D(\mathcal{F}(\mathcal{O}_{\Psi}))
\leq (1+\gamma(\Psi)) R_D(\mathcal{F}(\mathcal{O})).
\end{eqnarray}

\end{thm}
The proof of Theorem \ref{thm:main1} is presented in Appendix \ref{proof:main1}.
Theorem \ref{thm:main1} tells us that with access to the resource channel, the Rademacher complexity is bounded by
the free robustness of the channel.

Next, let us consider the case where the resource channel can be used multiple times. In this case, the relevant function class is that defined by Eq.~\eqref{eq:defOpsiK}. By Eq.~\eqref{eq:hierarchySets}, the following relationship follows immediately

\begin{eqnarray}
\hat{R}_S(\mathcal{O}^{(k)}_{\Psi})
\leq \hat{R}_S(\mathcal{O}^{(k+1)}_{\Psi}).
\end{eqnarray}

\begin{thm}\label{thm:main2}
Given $m$ independent samples $S=(\vec{z}_1,\ldots,\vec{z}_m)$ and a resource channel 
 $\Psi$, we have the following bound 
\begin{align}
\hat{R}_S(\mathcal{F}(\mathcal{O}^{(k)}_{\Psi}))
\leq \gamma^* \hat{R}_S(\mathcal{F}(\mathcal{O})),
\end{align}
where $\gamma^*=\min\set{1+2\gamma_{\max,n}, (1+2\gamma(\Psi))^k}$, and $\gamma_{\max,n}$ is the maximal free robustness over quantum channels on $n$ qubits. Therefore for any probability distribution $D$ on the sample space, with each sample $\vec{z}_i$ chosen independently 
according to $D$ for $i=1,\ldots,m$, we have 
\begin{align}
R_D(\mathcal{F}(\mathcal{O}^{(k)}_{\Psi}))
\leq \gamma^*R_D(\mathcal{F}(\mathcal{O}))
.
\end{align}

\end{thm}
The proof of Theorem \ref{thm:main2} is presented in Appendix \ref{apen:prof_main2}.
Theorem \ref{thm:main2} tells us that the Rademacher complexity for the case where we have access to multiple copies of a resource channel
has an upper bound that depends on the Rademacher complexity for the case where there is no resource channel, the free robustness of the resource channel $\Psi$ and the number of times
$\Psi$ is used.

Now, let us give some examples to illustrate our results.

\noindent\textit{\textbf{Example 1}}:
Consider quantum circuits whose gates all belong to the Clifford group, and denote the set of Clifford channels associated with such circuits by 
 $\mathcal{STAB}$. Of interest to us is the Rademacher complexity of  $\mathcal F(\mathcal{STAB})$ with respect to $m$
independent samples $S=(\vec{z}_i)^m_{i=1}$, denoted by 
$\hat{R}_S(\mathcal{F}(\mathcal{STAB}))$.
As we shall show in Appendix \ref{apen:prof_exam}, we get the following bound for stabilizer circuits
\begin{eqnarray}
\hat{R}_S(\mathcal{F}(\mathcal{STAB}))
\leq 4\frac{(1+o(1))n}{\sqrt{m}} \max_{\Phi\in \mathcal{STAB}}\norm{\vec{f}_{\Phi}}_{\infty},
\end{eqnarray}
where $\vec{f}_{\Phi}=(f_{\Phi}(\vec{z}_i))^m_{i=1}$.
Now, while such circuits
can be efficiently simulated on a classical computer, by the Gottesman-Knill theorem \cite{gottesman1997heisenberg},
circuits formed from the
Clifford+$T$ universal gate set, where
$T=\diag[1,e^{i\pi/4}]$, are believed to preclude efficient classical simulation \cite{terhal2004adaptive,nest2010classical}. This motivates us to consider quantum circuits consisting of both Clifford gates and the $T$ gate. We define the set $\mathcal{STAB}^{(k)}_{T}$ to be the set of quantum channels formed from Clifford unitaries and at most $k$ $T$ gates. As we shall show in Appendix \ref{apen:prof_exam}, the following 
upper bound holds for the Rademacher complexity of $\mathcal{STAB}^{(k)}_{T}$:
\begin{eqnarray}
\nonumber \hat{R}_S(\mathcal{F}(\mathcal{STAB}^{(k)}_{T}))
&\leq&  \left(1+\sqrt{2}\right)^k\hat{R}_S(\mathcal{F}(\mathcal{STAB}))\\
&\leq& O\left(\left(1+\sqrt{2}\right)^k\frac{n}{\sqrt{m}}\right)\max_{\Phi\in \mathcal{STAB}}\norm{\vec{f}_{\Phi}}_{\infty},
\nonumber\\
\end{eqnarray}
where we used the fact that the free robustness of the $T$ gate is upper bounded by $\sqrt{2}/2$.

\noindent\textit{\textbf{Example 2}}:
Consider the instantaneous quantum polynomial-time (IQP) circuits, a restricted model of quantum computation that has been proposed as a candidate for demonstrating quantum computational supremacy in the near term \cite{BremnerPRL2016,Bremner2017quantum,bremner2010classical,dalzell2020many}.
The structure of IQP circuits is quite simple: each circuit
has the form $H^{\ot n} DH^{\ot n}$, where
$D$ is a subcircuit with gates chosen from 
$\set{Z, CZ, CCZ}$ (see Fig.~\ref{fig2}).
Let us define $\mathcal{I}$ to be the set of IQP circuits
for which the gates in $D$ are from the gate set $\set{Z, CZ}$ and which contains at least one $CZ$ gate (the case in which the circuits do not contain a $CZ$ gate is trivial). Note that each circuit in $\mathcal I$ is also a Clifford circuit.
Moreover, $\mathcal{I}$ is a finite set, and the size of $\mathcal{I}$ is $O(2^{n^2})$. 
Thus, we have the following bound 
\begin{eqnarray}
\hat{R}_S(\mathcal{F}(\mathcal{I}))\leq \frac{O(n)}{\sqrt{m}}\max_{\Phi\in \mathcal{I}}\norm{\vec{f}_{\Phi}}_{\infty}
.
\end{eqnarray}
While $\mathcal{I}$ can be 
efficiently simulated on a classical computer \cite{gottesman1997heisenberg}, IQP circuits formed from the gate set $\mathcal{I}+CCZ$ 
are hard to simulate classically \cite{BremnerPRL2016,Bremner2017quantum,bremner2010classical}, which motivates us to consider the set $\mathcal{I}^{(k)}_{CCZ}$ of IQP circuits with at most 
$k$ CCZ gates. As we shall show in Appendix \ref{apen:prof_exam}, the following bound holds:
\begin{eqnarray}
\hat{R}_S(\mathcal{F}(\mathcal{I}^{(k)}_{CCZ}))\leq \frac{O((n^2+k\log n)^{1/2})}{\sqrt{m}}\max_{\Phi\in \mathcal{I}^{(k)}_{CCZ}}\norm{\vec{f}_{\Phi}}_{\infty}.
\end{eqnarray}
.

\begin{figure}[!h]
  \center{\includegraphics[width=6cm]  {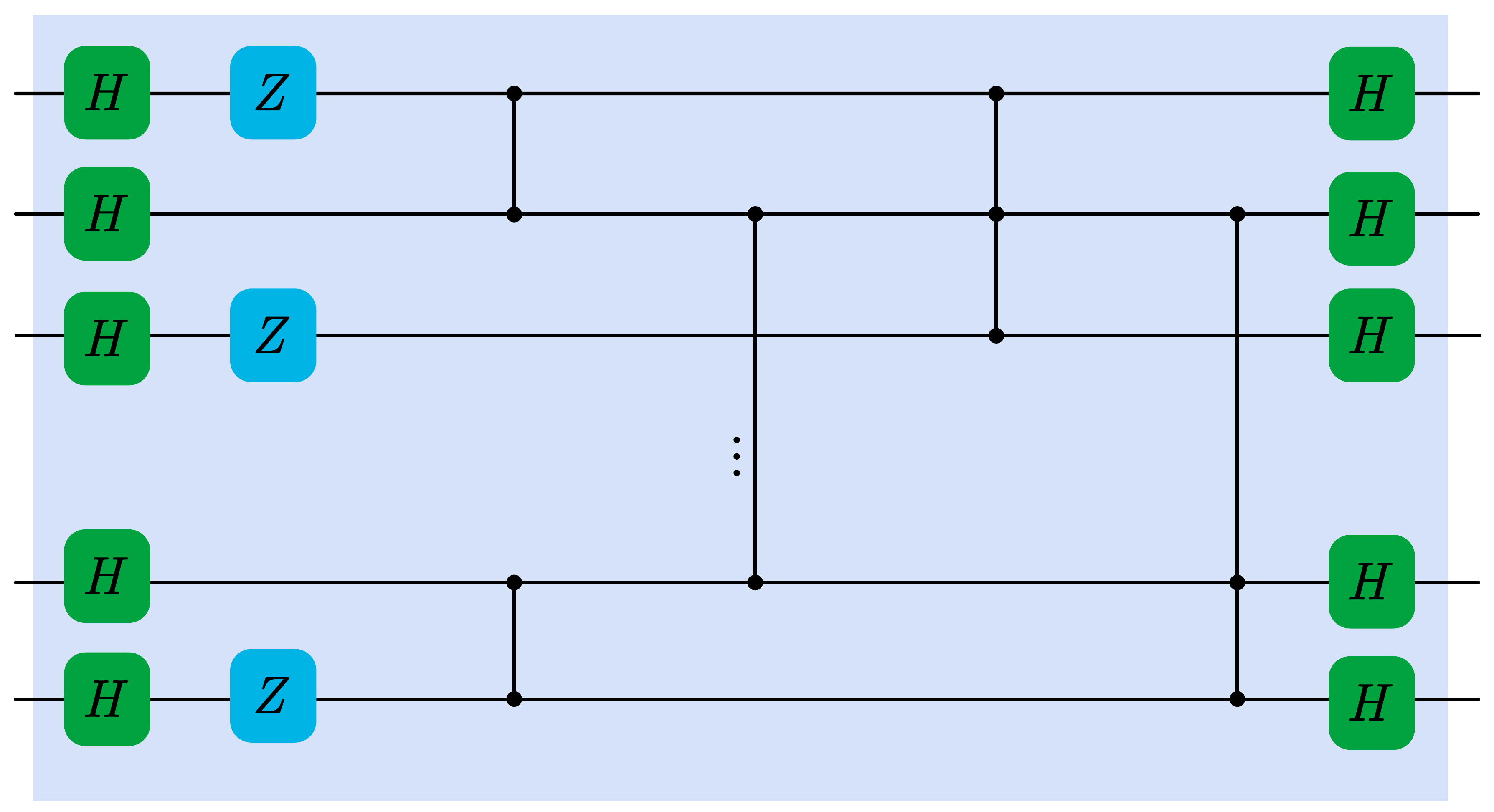}}     
  \caption{An example of an IQP circuit, which has the form $H^{\ot n}DH^{\ot n}$, where the gates in $D$ may be chosen only from 
  the gate set $\set{Z, CZ, CCZ}$.}
  \label{fig2}
 \end{figure}

\subsection{Generalization error bounds}

Given a sample $\vec{z}=(\vec{x}, \vec{y})$ (e.g., $\vec{y}=g(\vec{x})$ for some unknown function $g$),
let us consider the 
loss function 
$l(\vec{z}_i, \Phi)=1-f_{\Phi}(\vec{z}_i)$ where
$f_{\Phi}$ is defined by Eq.~\eqref{def:fun}. 
Then the expected error with respect to some unknown probability 
distribution $D$ on $Z^n_2\times Z^n_2$ is 
\begin{eqnarray}
er_D(\Phi)=\mathbb{E}_{\vec{z}\sim D}
l(\vec{z}, \Phi)
\end{eqnarray}
Given $m$ independent samples $S=(\vec{z}_1,\ldots,\vec{z}_m)$, the empirical error
is 
\begin{eqnarray}
er_S(\Phi)
=\frac{1}{m}
\sum_i l(\vec{z}_i,\Phi).
\end{eqnarray}
The difference between $er_S$ and $er_D$ is called the \textit{generalization error}, which determines the performance of the  function 
$f$ on the unseen data drawn from the unknown probability distribution. 
The Rademacher complexity provides a bound on the generalization error by the following result.

\begin{lem}[\cite{Bartlett03}]
If the loss function $l(f(\vec{x}),\vec{y})$ takes values in $[0,B]$, then 
 for any $\delta>0$, the following statement holds
for any function $f\in \mathcal{F}$ with probability at least $1-\delta$:
\begin{eqnarray*}
er_D(f)\leq er_S(f)
+2B \hat{R}_S(l_{\mathcal{F}})
+3B\sqrt{\frac{\log(2/\delta)}{2m}}
\end{eqnarray*}
 where the function class $l_{\mathcal{F}}:=\set{l_f: (\vec{x},\vec{y})\to l(f(\vec{x}),\vec{y})| f\in \mathcal{F}} $,
and  $\hat{R}_S(l_{\mathcal{F}})$ is the Rademacher complexity of the function class $l_{\mathcal{F}}$ on the $m$ given samples 
$S=\set{(\vec{x}_i, \vec{y}_i)}^m_{i=1}$. 
 
\end{lem}

Using this result and Theorem \ref{thm:main2}, we get the following upper bound on the generalization error for the function class $\mathcal{F}(\mathcal{O}^{(k)}_{\Psi})$ in terms of the 
Rademacher complexity of $\mathcal{F}(\mathcal{O})$ and $\gamma^*$.
\begin{prop}
Consider a set of quantum circuits $\mathcal{O}$ and let $\Psi\notin \mathcal{O}$. For any $\delta>0$, the following statement holds
for all  $\Phi\in \mathcal{O}^{(k)}_{\Psi}$ with probability at least $1-\delta$
\begin{eqnarray}
\nonumber er_D(\Phi)\leq er_S(\Phi)
+2\gamma^* \hat{R}_S(\mathcal{F}(\mathcal{O}))
+3\sqrt{\frac{\log(2/\delta)}{2m}},
\end{eqnarray}
where $\gamma^*= \min\set{(1+2\gamma(\Psi))^k , 1+2\gamma_{\max,n}}$.
\end{prop}

\section{Conclusion}
In this paper, we investigated the effects of quantum resources on the statistical complexity of quantum circuits. We considered the Rademacher and Gaussian complexities of the quantum-circuit-based learning model in 
two cases: (1) quantum circuits with access to only a restricted set of channels $\mathcal O$, and (2) quantum circuits with access to a restricted set of channels $\mathcal O$ together with an additional resource channel $\Psi \in \mathcal O$.
 We show that by adding a resource channel to a set of free channels, the Radamacher and Gaussian complexities are increased by an amount that is bounded by the free robustness of the resource channel multiplied by the number of times the channel is used. 
 We applied our results to two special cases: (1) stabilizer circuits that are supplemented with a limited number of T gates and (2) instantaneous quantum polynomial-time Clifford circuits that are supplemented with a limited number of CCZ gates.
 Using this result, we derive an upper bound on the generalization error associated with learning from the training data arising from such circuits.

Our results reveal  a new connection between quantum resources and the statistical complexity of quantum circuits, 
which paves the way for further research into the statistical complexity of learning models based on quantum circuits, like the variational quantum eigensolver and the quantum neural network. Furthermore, from a quantum resource theoretic point of view, 
our results also provide a new operational interpretation of free robustness in general resource theories.

While we focused on the quantum circuit model in this paper, it will be interesting to generalize
our results to other computational models such as measurement-based quantum computation (MBQC), tensor networks, etc.  
Besides the Rademacher and Gaussian complexities, there  are also other measures of statistical complexity of function classes, such as the metric entropy, the VC dimension (or more generally, the pseudo-dimension \cite{caro2020pseudo}), and the topological entropy \cite{bu2020depth}. It will be interesting to see 
the effects of quantum resources using these other measures of statistical complexity. We leave this problem for further research.

\begin{acknowledgments}
K. B. thanks Arthur Jaffe  and Zhengwei Liu for the help and support during the breakout of the COVID-19 pandemic. 
K. B. acknowledges the support of
 ARO Grants W911NF-19-1-0302 and
W911NF-20-1-0082, and the support from Yau Mathematical
Science Center at Tsinghua University during the visit. 

\end{acknowledgments}

 \bibliography{SatCom-lit}

\appendix
\widetext

\section{Basic properties of Rademacher complexity}
In this section, we list several basic properties of the Rademacher complexity,
which may be found in  \cite{shalev2014}.

Given a subset  $A$ of $\real^m$, the Rademacher complexity of 
$A$ is defined as
\begin{eqnarray}
\hat{R}(A)=\mathbb{E}_{\vec{\epsilon}}
\sup_{\vec{v}\in A}
\frac{1}{m}\sum_{i}\epsilon_iv_i,
\end{eqnarray}
where $\set{\epsilon_i}_i$ are i.i.d Rademacher random variables.

\begin{prop}[\cite{Bartlett03,shalev2014}]
The Rademacher complexity satisfies the following properties:

(1)
\begin{eqnarray}
\hat{R}(A)=\hat{R}(\mathrm{Conv}(A)),
\end{eqnarray}
where
$\mathrm{Conv}(A)=\left\{\sum_i\lambda_i\vec{v}_i: \vec{v}_i\in A, \lambda_i\geq 0, \sum_i\lambda_i=1\right\}$.

(2) For any $c\in \real$, we have 
\begin{eqnarray}
\hat{R}(cA)=|c|\hat{R}(A),
\end{eqnarray}
where $cA:=\set{c\vec{v}:\vec{v}\in A}$.

(3) For any $\vec{c}\in \real^m$, we have
\begin{eqnarray}
\hat{R}(A+\vec{c})=\hat{R}(A),
\end{eqnarray}
where $A+\vec c:=\set{\vec{v}+\vec{c}:\vec{v}\in A}$.

(4) For any $A_1,A_2\subset \real^m$, we have
\begin{eqnarray}
\hat{R}(A_1+A_2)
=\hat{R}(A_1)+\hat{R}(A_2),
\end{eqnarray}
where $A_1+A_2:=\set{\vec{v}_1+\vec{v}_2: \vec{v}_1\in A_1,\vec{v}_2\in A_2}$.

(5) Given a Lipschitz function $\phi:\real\to \real$ with Lipschitz constant $L$ and
$\phi(0)=0$, we have 
\begin{eqnarray}
\hat{R}(\phi\circ A)
\leq L\hat{R}(A),
\end{eqnarray}
where $\phi\circ A:=\set{(\phi(x_1), \phi(x_2),...,\phi(x_m)): (x_1, x_2, ...,x_m)\in A}$. 

\end{prop}

When the set $A$ is finite, Massart's lemma gives an upper bound for the Rademacher complexity of $A$. 

\begin{lem}[Massart's lemma \cite{shalev2014}]\label{lem:massart}
Given a finite set $A\subset \real^m$, then we have 
\begin{eqnarray}
\hat{R}(A)\leq \max_{\vec{v}\in A}\norm{\vec{v}}_2\frac{\sqrt{2\log |A|}}{m},
\end{eqnarray}
where  $|A|$ denotes the size of the finite set $A$. 
\end{lem}

We now state an important result of the Rademacher complexity, which allows it to be estimated from a single sample set $S=(z_1,\ldots,z_m)$. 
\begin{lem}[\cite{Bartlett03}]
Let $\mathcal G$ be a set of functions $\mathcal X \rightarrow [a,b]$, where $a< b$. Let $m \in \mathbb Z^+$ be a positive integer and $D$ be a probability distribution. Let $t>0$. Then,
\begin{align}
& \mathop{\mathrm{Pr}}_{
\substack{(z_1,\ldots,z_m)\sim D^m \\
\epsilon_1,\ldots,\epsilon_m \sim \mathrm{Rad}
}
}
\left[\left|R_D(\mathcal{G})-\frac{1}{m}\sup_{g\in\mathcal{G}}\sum_{i=1}^m \epsilon_i
g(z_i)\right|\geq t
\right]
\nonumber\\
&\hspace{10ex}
\leq
2\exp\left[- \frac{2m t^2}{(b-a)^2+4\max\set{a^2,b^2}}\right],
\end{align}
and
\begin{align}
\mathop{\mathrm{Pr}}_{S\sim D^m}
\left[
\left|R_D(\mathcal{G})-\hat{R}_S(\mathcal{G})
\right|\geq t\right]
\leq
2\e^{-2m t^2 /(b-a)^2}.
\end{align}
\end{lem}

\section{Proof of Theorem 1}\label{proof:main1}
\begin{proof}

First, let us rewrite the Rademacher complexity as follows
\begin{eqnarray}
\hat{R}_S(\mathcal{F})=\mathbb{E}
\sup_{f\in \mathcal{F}}
\frac{1}{m}\sum^m_{i=1}\epsilon_if(\vec{z}_i)
=\mathbb{E}_{\vec{\epsilon}}
\sup_{f\in\mathcal{F}}
\inner{\vec{\epsilon}}{\vec{f}},
\end{eqnarray}
where $\vec{\epsilon}=(\epsilon_1,...,\epsilon_m)\in\set{\pm}^n$,
$\vec{f}=(f(\vec{z}_1),...,f(\vec{z}_m))$ and 
$
\inner{\vec{\epsilon}}{\vec{f}}
=\frac{1}{m}\sum^m_{i=1}
\epsilon_if(\vec{z}_i)
$.

The  inequality $\hat{R}_S(\mathcal{F}(\mathcal{O}))
\leq
\hat{R}_S(\mathcal{F}(\mathcal{O}_{\Psi}))$ comes directly from the definition of Rademacher complexity and the fact that 
$\mathcal{O}\subset \mathcal{O}_{\Psi}$.
Hence,  we only need to prove that
\begin{eqnarray} 
\hat{R}_S(\mathcal{F}(\mathcal{O}_{\Psi}))
\leq (1+\gamma(\Psi)) \hat{R}_S(\mathcal{F}(\mathcal{O})).
\end{eqnarray}

Let us define the set $A$ as follows
\begin{eqnarray}\label{def:A}
A=\left\{\vec{\epsilon}\in\set{\pm1}^m \Bigg|
\inner{\vec{\epsilon}}{\vec{f}_{\Psi}}>\sup_{ f\in\mathcal{F}(\mathcal{O})} \inner{\vec{\epsilon}}{\vec{f}} 
\right\}.
\end{eqnarray}

To finish the proof, we need the following two  lemmas about the basic properties of the set 
$A$ defined in \eqref{def:A}.

\begin{lem}\label{lem:eq}
Given the set A defined in \eqref{def:A}, we have
\begin{eqnarray}
A\cap (-A)=\emptyset,
\end{eqnarray}
where $-A:=\set{-\vec{\epsilon}| \vec{\epsilon}\in A}$.
\end{lem}
\begin{proof}
Based on the definition of the set $A$, we have 
\begin{eqnarray*}
\inner{\vec{\epsilon}}{\vec{f}_{\Psi}}>\sup_{ f\in\mathcal{F}(\mathcal{O})} \inner{\vec{\epsilon}}{\vec{f}},
\end{eqnarray*}
for any $\vec{\epsilon}\in A$. Thus, we have
\begin{eqnarray*}
\inner{-\vec{\epsilon}}{\vec{f}_{\Psi}}
<-\sup_{ f\in\mathcal{F}(\mathcal{O})} \inner{\vec{\epsilon}}{\vec{f}} 
=\inf_{ f\in\mathcal{F}(\mathcal{O})} \inner{-\vec{\epsilon}}{\vec{f}} 
\leq \sup_{ f\in\mathcal{F}(\mathcal{O})} \inner{-\vec{\epsilon}}{\vec{f}},
\end{eqnarray*}
That is, $-\vec{\epsilon}\in A^c$. Therefore, we have
$A\cap A^c=\emptyset$.
\end{proof}

\begin{lem}\label{lem:ineq}
Given the set A defined in Eq.~\eqref{def:A}, 
we have 
\begin{eqnarray}
\sum_{\vec{\epsilon}\in A}
\sup_{f\in \mathcal{F}(\mathcal{O})}
\inner{\vec{\epsilon}}{\vec{f}}
+\sum_{\vec{\epsilon}\in -A}
\sup_{f\in \mathcal{F}(\mathcal{O})}
\inner{\vec{\epsilon}}{\vec{f}}
\leq\sum_{\vec{\epsilon}\in \set{\pm1}^m}
\sup_{f\in \mathcal{F}(\mathcal{O})}
\inner{\vec{\epsilon}}{\vec{f}}.
\end{eqnarray}
\end{lem}
\begin{proof}
First, due to Lemma \ref{lem:eq}, we have 
\begin{eqnarray*}
\sum_{\vec{\epsilon}\in A}
\sup_{f\in \mathcal{F}(\mathcal{O})}
\inner{\vec{\epsilon}}{\vec{f}}
+\sum_{\vec{\epsilon}\in -A}
\sup_{f\in \mathcal{F}(\mathcal{O})}
=\sum_{\vec{\epsilon}\in A\cup (-A)}
\sup_{f\in \mathcal{F}(\mathcal{O})}
\inner{\vec{\epsilon}}{\vec{f}}.
\end{eqnarray*}
Hence, we only need to prove that
\begin{eqnarray*}
\sum_{\vec{\epsilon}\in A\cup (-A)}
\sup_{f\in \mathcal{F}(\mathcal{O})}
\inner{\vec{\epsilon}}{\vec{f}}
\leq\sum_{\vec{\epsilon}\in \set{\pm1}^m}
\sup_{f\in \mathcal{F}(\mathcal{O})}
\inner{\vec{\epsilon}}{\vec{f}}.
\end{eqnarray*}
Let us define the set $B:=A\cup (-A)$, 
then it easy to verify that
$-\vec{\epsilon}\in B^c$, for any $\vec{\epsilon}\in B^c$.
Then, we have
\begin{eqnarray*}
\sum_{\vec{\epsilon}\in \set{\pm1}^m}
\sup_{f\in \mathcal{F}(\mathcal{O})}
\inner{\vec{\epsilon}}{\vec{f}}-\sum_{\vec{\epsilon}\in A\cup (-A)}
\sup_{f\in \mathcal{F}(\mathcal{O})}
\inner{\vec{\epsilon}}{\vec{f}}
=\sum_{\vec{\epsilon}\in B^c}
\sup_{f\in \mathcal{F}(\mathcal{O})}
\inner{\vec{\epsilon}}{\vec{f}}
\geq \sum_{\vec{\epsilon}\in B^c}
\inner{\vec{\epsilon}}{\vec{f}}
=
\inner{ \sum_{\vec{\epsilon}\in B^c}\vec{\epsilon}}{\vec{f}}
=0.
\end{eqnarray*}

\end{proof}

Based on the definition of free robustness, there exist channels $\Phi_1,\Phi_2\in \mathrm{Conv}(\mathcal{O})$ such that 
\begin{eqnarray*}
\Psi=\left(1+\gamma(\Psi)\right)\Phi_1-\gamma(\Psi)\Phi_2.
\end{eqnarray*} 
Due to the linearity of function $f_{\Phi}$ with respect to $\Phi$, we have 
\begin{eqnarray*}
f_{\Psi}
=(1+\gamma(\Psi))f_{\Phi_1}-\gamma(\Psi)f_{\Phi_2}.
\end{eqnarray*}

Therefore, 
\begin{eqnarray*}
\hat{R}_S(\mathcal{F}(\mathcal{O}_{\Psi}))
&=&\frac{1}{2^m}
\sum_{\vec{\epsilon}\in A}
\inner{\vec{\epsilon}}{\vec{f}_{\Psi}}
+\frac{1}{2^m}
\sum_{\vec{\epsilon}\in A^c}
\sup_{f\in \mathcal{F}(\mathcal{O})}
\inner{\vec{\epsilon}}{\vec{f}}\\
&=&\frac{1}{2^m}
\sum_{\vec{\epsilon}\in A}
[(1+\gamma(\Psi))\inner{\vec{\epsilon}}{f_{\Phi_1}}-\gamma(\Psi)\inner{\vec{\epsilon}}{f_{\Phi_2}}]
+\frac{1}{2^m}
\sum_{\vec{\epsilon}\in A^c}\sup_{f\in \mathcal{F}(\mathcal{O})}
\inner{\vec{\epsilon}}{\vec{f}}\\
&=&
\frac{1}{2^m}
\sum_{\vec{\epsilon}\in A}
\inner{\vec{\epsilon}}{f_{\Phi_1}}
+\frac{1}{2^m}
\sum_{\vec{\epsilon}\in A^c}\sup_{f\in \mathcal{F}(\mathcal{O})}
\inner{\vec{\epsilon}}{\vec{f}}
+\gamma(\Psi)\frac{1}{2^m}
\sum_{\vec{\epsilon}\in A}
[\inner{\vec{\epsilon}}{f_{\Phi_1}}-\inner{\vec{\epsilon}}{f_{\Phi_2}}]\\
&\leq & 
\frac{1}{2^m}
\sum_{\vec{\epsilon}\in A}
\sup_{f\in \mathcal{F}(\mathcal{O})}
\inner{\vec{\epsilon}}{\vec{f}}
+\frac{1}{2^m}
\sum_{\vec{\epsilon}\in A^c}\sup_{f\in \mathcal{F}(\mathcal{O})}
\inner{\vec{\epsilon}}{\vec{f}}
+\gamma(\Psi)\frac{1}{2^m}
\sum_{\vec{\epsilon}\in A}
[\inner{\vec{\epsilon}}{f_{\Phi_1}}-\inner{\vec{\epsilon}}{f_{\Phi_2}}]\\
&=&
R_S(\mathcal{F}(\mathcal{O}))
+\gamma(\Psi)\frac{1}{2^m}
\sum_{\vec{\epsilon}\in A}
[\inner{\vec{\epsilon}}{f_{\Phi_1}}-\inner{\vec{\epsilon}}{f_{\Phi_2}}]\\
&= &
\hat{R}_S(\mathcal{F}(\mathcal{O}))
+\gamma(\Psi)\frac{1}{2^m}
\left[\sum_{\vec{\epsilon}\in A}
\inner{\vec{\epsilon}}{f_{\Phi_1}}+\sum_{\vec{\epsilon}\in -A}\inner{\vec{\epsilon}}{f_{\Phi_2}}\right]\\
&\leq&\hat{R}_S(\mathcal{F}(\mathcal{O}))
+\gamma(\Psi)\frac{1}{2^m}
\left[\sum_{\vec{\epsilon}\in A}
\sup_{f\in \mathcal{F}(\mathcal{O})}\inner{\vec{\epsilon}}{f}+\sum_{\vec{\epsilon}\in -A}\sup_{f\in \mathcal{F}(\mathcal{O})}\inner{\vec{\epsilon}}{f}\right]\\
&\leq&\hat{R}_S(\mathcal{F}(\mathcal{O}))
+\gamma(\Psi)\frac{1}{2^m}
\sum_{\vec{\epsilon}\in \set{\pm 1}^m}\sup_{f\in \mathcal{F}(\mathcal{O})}\inner{\vec{\epsilon}}{f}\\
&=&(1+\gamma(\Psi))\hat{R}_S(\mathcal{F}(\mathcal{O})),
\end{eqnarray*}
where the first and second inequality comes from the fact that $\Phi_1,\Phi_2\in \mathrm{Conv}(\mathcal{O})$, and the last inequality comes from Lemma \ref{lem:ineq}.

\end{proof}

\section{Proof of Theorem 2}\label{apen:prof_main2}

First,
let us prove the  following lemma about the relationship between the Rademacher complexities of 
$\mathcal{O}^{(k+1)}_{\Psi}$ and $\mathcal{O}^{(k)}_{\Psi}$.
\begin{lem}\label{lem:rel_k}
Given $m$ independent samples $S=(\vec{z}_1,\ldots ,\vec{z}_m)$ and a resource channel 
 $\Psi$, we have 
 \begin{eqnarray}
 \hat{R}_S\left(\mathcal{F}(\mathcal{O}^{(k+1)}_{\Psi})\right)
 \leq(1+2\gamma(\Psi))\hat{R}_S\left(\mathcal{F}(\mathcal{O}^{(k)}_{\Psi})\right),
 \end{eqnarray}
 for any $k\geq 0$.
\end{lem}
\begin{proof}
By the definition 
of free robustness, there exist channels $\Phi_1,\Phi_2\in \mathrm{Conv}(\mathcal{O})$ such that 
\begin{eqnarray*}
\Psi=(1+\gamma(\Psi))\Phi_1-\gamma(\Psi)\Phi_2.
\end{eqnarray*} 
Therefore, for 
any channel $\Phi\in  \mathcal{O}^{(k+1)}_{\Psi}$, there 
exist two channels $\Phi', \Phi''\in \mathrm{Conv}(\mathcal{O}^{(k)}_{\Psi})$ such that 
\begin{eqnarray*}
\Psi=(1+\gamma(\Psi))\Phi'-\gamma(\Psi)\Phi''.
\end{eqnarray*}
Therefore, 
\begin{eqnarray*}
\mathcal{O}^{(k+1)}_{\Psi}
\subset 
(1+\gamma(\Psi))\mathrm{Conv}(\mathcal{O}^{(k)}_{\Psi})
-\gamma(\Psi)\mathrm{Conv}(\mathcal{O}^{(k)}_{\Psi}).
\end{eqnarray*}
Therefore, 
\begin{eqnarray*}
\hat{R}_S\left(\mathcal{F}(\mathcal{O}^{(k+1)}_{\Psi})\right)
&\leq& \hat{R}_S\left[\mathcal{F}\left((1+\gamma(\Psi))\mathrm{Conv}(\mathcal{O}^{(k)}_{\Psi})
-\gamma(\Psi)\mathrm{Conv}(\mathcal{O}^{(k)}_{\Psi})\right)\right]\\
&=&(1+\gamma(\Psi))\hat{R}_S\left(\mathcal{F}(\mathcal{O}^{(k)}_{\Psi})\right)+
\gamma(\Psi)\hat{R}_S\left(\mathcal{F}(\mathcal{O}^{(k)}_{\Psi})\right)\\
&=&(1+2\gamma(\Psi))\hat{R}_S\left(\mathcal{F}(\mathcal{O}^{(k)}_{\Psi})\right),
\end{eqnarray*}
where the first equality comes from that fact that 
$\hat{R}_S(\sum_i \mathcal{F}_i)=\sum_i\hat{R}_S(\mathcal{F}_i)$ where each $\mathcal{F}_i$ is a function class, and  the facts that 
Rademacher complexity is invariant under convex combination and $\hat{R}_S(c\mathcal{F})=|c|\hat{R}_s(\mathcal{F})$.
\end{proof}

We are now ready to prove the lemma.
\begin{proof}
Based on Lemma \ref{lem:rel_k}, we have the following inequality:
\begin{eqnarray*}
\hat{R}_S(\mathcal{F}\left(\mathcal{O}^{(k+1)}_{\Psi})\right)
\leq (1+2\gamma(\Psi))^k \hat{R}_S\left(\mathcal{F}(\mathcal{O})\right).
\end{eqnarray*}

Besides, for any $\Phi\in \mathcal{O}^{(k)}_{\Psi}$, there exist 
$\Phi_1,\Phi_2\in\mathcal{O}$ such that 
\begin{eqnarray}
\Phi=(1+\gamma)\Phi_1-\gamma\Phi_2,
\end{eqnarray}
where $\gamma\leq\gamma_{\max,n}$. Therefore, we have 
\begin{eqnarray*}
 \mathcal{O}^{(k)}_{\Psi}
 \subset (1+\gamma_{\max,n})\mathrm{Conv}(\mathcal{O})-\gamma_{\max,n}\mathrm{Conv}(\mathcal{O}),
\end{eqnarray*}
for any integer $k$.
Hence, we have
\begin{eqnarray*}
\hat{R}_S\left(\mathcal{F}(\mathcal{O}^{(k)}_{\Psi})\right)
&\leq& \hat{R}_S[\mathcal{F}\left((1+\gamma_{\max,n})\mathrm{Conv}(\mathcal{O})
-\gamma_{\max,n}\mathrm{Conv}(\mathcal{O})\right)]\\
&=&(1+\gamma_{\max,n})\hat{R}_S(\mathcal{F}(\mathcal{O}))+
\gamma_{\max,n}\hat{R}_S(\mathcal{F}(\mathcal{O}))\\
&=&(1+2\gamma_{\max,n})\hat{R}_S(\mathcal{F}(\mathcal{O})),
\end{eqnarray*}
for any integer $k$.
Therefore, we have 
\begin{eqnarray*}
\hat{R}_S\left(\mathcal{F}(\mathcal{O}^{(k)}_{\Psi})\right)
\leq \min\left\{1+2\gamma_{\max,n}, (1+2\gamma(\Psi))^k\right\}
\hat{R}_S(\mathcal{F}(\mathcal{O})).
\end{eqnarray*}

\end{proof}

\section{Proof of Example 1 and 2}\label{apen:prof_exam}

\subsection{Example 1: Proofs}
By Choi's representation of quantum channels \cite{watrous2018theory}, 
the function $f_{\Phi}$ can be written as follows
\begin{eqnarray}
f_{\Phi}(\vec{x},\vec{y})=2^n\Tr{\Phi\ot \mathbb{I}(\proj{\Lambda})\proj{\vec{x}}\ot \proj{\vec{y}}},
\end{eqnarray}
where $\ket{\Lambda}=1/\sqrt{2^n}\sum_{\vec{x}}\ket{\vec{x}}\ket{\vec{x}}$. 
Since $\Phi$ is a (unitary) stabilizer circuit and
$\ket{\Lambda}$ is a pure stabilizer state,
$\Phi\ot \mathbb{I}(\proj{\Lambda})$ is
a stabilizer state on $2n$ qubits.
Since the number
of pure stabilizer states on $2n$ qubits is $2^{(0.5+o(1))(2n)^2}$ \cite{Scott2004}, let us consider the vector $\vec{f}_{\Phi}=(f_{\Phi}(\vec{z}_i))^m_{i=1}$,
where
the set $\left\{\vec{f}_{\Phi}\right\}_{\Phi\in \mathcal{STAB}}$ is a finite set satisfying
\begin{eqnarray}
\left|\left\{\vec{f}_{\Phi}\right\}_{\Phi\in \mathcal{STAB}}\right|\leq 2^{(0.5+o(1))(2n)^2}. 
\end{eqnarray}

Hence, we have 
\begin{eqnarray*}
R_S(\mathcal{F}(\mathcal{STAB}))
&\leq& 4\frac{(1+o(1))n}{m}\max_{\Phi\in \mathcal{STAB}}\norm{\vec{f}_{\Phi}}_2\\
&\leq& 4\frac{(1+o(1))n}{\sqrt{m}}\max_{\Phi\in \mathcal{STAB}}\norm{\vec{f}_{\Phi}}_{\infty},
\end{eqnarray*}
where  the first inequality comes from Massart's Lemma (see Lemma \ref{lem:massart})
and the second inequality comes from the fact that $\norm{\cdot}_2\leq \sqrt{m}\norm{\cdot}_{\infty}$. 

Now, let us assume that we have 
access to the $T$ gate. In this case, let us define the corresponding sets of quantum channels $ \mathcal{STAB}_{T}$ and
$\mathcal{STAB}^{(k)}_{T}$. 
The free robustness of $T$ gate is
$
\gamma(T)\leq \sqrt{2}/2
$ (we leave open the question about whether this bound is tight) as the  $T$ gate written as a quantum channel $\Phi_T$ may be decomposed as follows
\begin{eqnarray}
\Phi_T(\cdot)=\left(\frac{1}{2}+\frac{\sqrt{2}}{2}\right)\Phi_S(\cdot)+\frac{1}{2}\Phi_Z(\cdot)-\frac{\sqrt{2}}{2}\Phi_{SZ}(\cdot),
\end{eqnarray}
where 
 $S=\mathrm{diag}[1,i]$ is the phase gate and 
$Z=\mathrm{diag}[1,-1]$ is the Pauli $Z$ gate .
By Theorem \ref{thm:main2}, we get the following
upper bound on the Rademacher complexity of 
$\mathcal{STAB}^{(k)}_{T}$:
\begin{eqnarray*}
\hat{R}_S(\mathcal{F}(\mathcal{STAB}_{T}))
\leq \left(1+\sqrt{2}/2\right)\hat{R}_S(\mathcal{F}(\mathcal{STAB}))
\leq O\left(\left(1+\sqrt{2}/2\right)\frac{n}{\sqrt{m}}\right)\max_{\Phi\in \mathcal{STAB}}\norm{\vec{f}_{\Phi}}_{\infty}.
\end{eqnarray*}
\begin{eqnarray*}
\hat{R}_S(\mathcal{F}(\mathcal{STAB}^{(k)}_{T}))
\leq \left(1+\sqrt{2}\right)^k\hat{R}_S(\mathcal{F}(\mathcal{STAB}))
\leq O\left(\left(1+\sqrt{2}\right)^k\frac{n}{\sqrt{m}}\right)\max_{\Phi\in \mathcal{STAB}}\norm{\vec{f}_{\Phi}}_{\infty}.
\end{eqnarray*}

\subsection{Example 2: Proofs}
Since $\mathcal{I}$ are the IQP circuits with only $Z$ and $CZ$ as internal gates, 
$\mathcal{I}$ is a finite set with size $|\mathcal{I}|=O(2^{n^2})$.
Hence
\begin{eqnarray*}
\hat{R}_S(\mathcal{F}(\mathcal{I}))\leq \frac{O(n)}{m}\max_{\Phi\in\mathcal{I}}\norm{\vec{f}_{\Phi}}_2
\leq \frac{O(n)}{\sqrt{m}}\max_{\Phi\in \mathcal{I}}\norm{\vec{f}_{\Phi}}_{\infty},
\end{eqnarray*}
where the first inequality comes from Massart's Lemma (see Lemma \ref{lem:massart})
and the second inequality comes from the fact that $\norm{\cdot}_2\leq \sqrt{m}\norm{\cdot}_{\infty}$. 

Now, let us consider IQP ciruits with access to the CCZ gate. 
Let us define $\mathcal{I}^{(k)}_{CCZ} $ to be the set of IQP circuits with at most 
$k$ CCZ gates. 
Then, the size of  $\mathcal{I}^{(k)}_{CCZ} $ is 
\begin{eqnarray}
\left|\mathcal{I}^{(k)}_{CCZ} \right|\leq |\mathcal{I}|\times\left(\sum^{k}_{j=0}\binom{n^3}{j}\right)\leq 2^{O(n^2)}n^{3k}.
\end{eqnarray}

Therefore, by Massart's Lemma,  we have
\begin{eqnarray}
\hat{R}_S\left(\mathcal{F}(\mathcal{I}^{(k)}_{CCZ})\right)
\leq\frac{O((n^2+k\log n)^{1/2})}{\sqrt{m}}\max_{\Phi\in \mathcal{I}^{(k)}_{CCZ}}\norm{\vec{f}_{\Phi}}_{\infty}.
\end{eqnarray}

\section{Results about the Gaussian complexity}
\label{sec:GaussianComplexity}
In the main text, we focused on the Rademacher complexity. In this appendix, we will show that similar results hold
for the Gaussian complexity. 
\begin{thm}
Given $m$ independent samples $S=(\vec{z}_1,...,\vec{z}_m)$ and a resource channel 
 $\Psi$, then we have the following bound 
\begin{eqnarray}
\hat{G}_S(\mathcal{F}(\mathcal{O}))
\leq
\hat{G}_S(\mathcal{F}(\mathcal{O}_{\Psi}))
\leq (1+\gamma(\Psi)) \hat{G}_S(\mathcal{F}(\mathcal{O})),
\end{eqnarray}
where $\gamma(\Psi)$ is the free robustness with respect to the set $\mathcal{O}$, that is,
\begin{eqnarray*}
\gamma(\Psi):=\min\left\{\lambda|
\frac{\Psi+\lambda\Phi}{1+\lambda}\in \mathrm{Conv}(\mathcal{O}), \Phi\in \mathrm{Conv}(\mathcal{O})\right\}.
\end{eqnarray*}
 Therefore, for any probability distribution $D$ on the sample space, if each sample $\vec{z}_i$ is chosen independently 
according to $D$ for $i=1,\ldots ,m$, then we have 

\begin{eqnarray}
G_D(\mathcal{F}(\mathcal{O}))
\leq
G_D(\mathcal{F}(\mathcal{O}_{\Psi}))
\leq (1+\gamma(\Psi)) R_D(\mathcal{F}(\mathcal{O})).
\end{eqnarray}

\end{thm}
\begin{proof}
The proof is the same as that for the Rademacher complexity, except that we will need to replace 
the set $A$ in Lemma \ref{lem:eq} and Lemma \ref{lem:ineq} by
\begin{eqnarray}
A'=\set{\vec{\epsilon}\in\real^m|
\inner{\vec{\epsilon}}{\vec{f}_{\Psi}}>\sup_{ f\in\mathcal{F}(\mathcal{O})} \inner{\vec{\epsilon}}{\vec{f}} 
}.
\end{eqnarray}
\end{proof}

\begin{thm}
Given $m$ independent samples $S=(\vec{z}_1,\ldots,\vec{z}_m)$ and a resource channel 
 $\Psi$, we have the following bound 
\begin{eqnarray}
\hat{G}_S(\mathcal{F}(\mathcal{O}^{(k)}_{\Psi}))
\leq \gamma^*\hat{G}_S(\mathcal{F}(\mathcal{O})),
\end{eqnarray}
where $\gamma^*=\min\set{1+2\gamma_{\max,n}, (1+2\gamma(\Psi))^k}$, and $\gamma_{\max,n}$ is the maximal free robustness over quantum channels on $n$ qubits. Given a probability distribution $D$ on the sample space, if each sample $\vec{z}_i$ chosen independently 
according to $D$ for $i=1,\ldots,m$, then we have 
\begin{eqnarray}
G_D(\mathcal{F}(\mathcal{O}^{(k)}_{\Psi}))
\leq \gamma^*G_D(\mathcal{F}(\mathcal{O})).
\end{eqnarray}

\end{thm}
\begin{proof}
This result also holds for the Gaussian complexity because the Gaussian 
complexity also satisfies convexity and invariance under convex combination.
\end{proof}

\section{Alternative definition of Rademacher and Gaussian complexity involving absolute values}
Given a set of real-valued functions $\mathcal{F}$, the Rademacher and Gaussian complexity with respect to a given 
sample $S=(z_1,....,z_m)$
may alternatively be defined as follows: 
\begin{eqnarray}
\bar{R}_S(\mathcal{F})=
\mathbb{E} \sup_{f\in \mathcal{F}}
\frac{1}{m}\left|\sum^m_{i=1}\epsilon_i f(z_i)\right|,
\end{eqnarray}
where the expectation is taken over the i.i.d Rademacher variables $\epsilon_1,\epsilon_2,...,\epsilon_m$, and
\begin{eqnarray}
\bar{G}_S(\mathcal{F})=
\mathbb{E} \sup_{f\in \mathcal{F}}
\frac{1}{m}\left|\sum^m_{i=1}g_i f(z_i)\right|,
\end{eqnarray}
where the expectation is taken over i.i.d random Gaussian variables with zero mean and variance 1, i.e., $g_i\sim \mathcal{N}(0,1)$.

As the only difference between  $\hat{R}_S$ and $\bar{R}_S$ is that the latter involves taking an absolute value $|\cdot|$, it is easy to see that $\bar{R}_S$ and $\bar{G}_S$
 satisfy
the following bounds. 
\begin{thm}
Given $m$ independent samples $S=(\vec{z}_1,...,\vec{z}_m)$ and a resource channel 
 $\Psi$, we have the following bound 
\begin{eqnarray}
\bar{R}_S(\mathcal{F}(\mathcal{O}))
\leq
\bar{R}_S(\mathcal{F}(\mathcal{O}^{(k)}_{\Psi}))
\leq \gamma^*\bar{R}_S(\mathcal{F}(\mathcal{O})),\
\bar{G}_S(\mathcal{F}(\mathcal{O}))
\leq
\bar{G}_S(\mathcal{F}(\mathcal{O}^{(k)}_{\Psi}))
\leq \gamma^* \bar{G}_S(\mathcal{F}(\mathcal{O})),
\end{eqnarray}
where $\gamma^*= \min\set{(1+2\gamma(\Psi))^k , 1+2\gamma_{\max,n}}$ and $\gamma_{\max,n}$ is the maximal free robustness over quantum channels on $n$ qubits. 
\end{thm}
The proof is the same as that of Theorem \ref{thm:main2}.

\end{document}